\documentclass[conference]{IEEEtran}
\usepackage{amsmath,amssymb,graphicx,epsfig, cite,algorithm,algorithmic,epstopdf, url}
\usepackage[utf8]{inputenc}
\usepackage[mathscr]{euscript}
\usepackage{color}
\usepackage{bm}
\usepackage{amsthm}
\usepackage[keeplastbox]{flushend}

\def\minwrt[#1]{\underset{#1}{\text{minimize }}}
\def\argminwrt[#1]{\underset{#1}{\text{arg min }}}
\def\argmaxwrt[#1]{\underset{#1}{\text{arg max }}}
\def\maxwrt[#1]{\underset{#1}{\text{maximize }}}
\def\maxemphwrt[#1]{\underset{#1}{\text{\emph{maximize} }}}

\newtheorem{proposition}{Proposition}

\newtheorem{definition}{Definition}

\newcommand{\abs}[1]{\left|#1\right|}

\def\by{{\bf y}}

\def\ccI{{\mathcal{I}}}

        \makeatletter
        \def\fps@eqnfloat{!t}
        \def\ftype@eqnfloat{4}
        
        \newenvironment{eqnfloat*}
               {\@dblfloat{eqnfloat}}
               {\end@dblfloat}
        \makeatother
\renewenvironment{thebibliography}[1]{%
 \begin{oldthebibliography}{#1}%
 \setlength{\parskip}{0ex}%
 \setlength{\itemsep}{0ex}%
}%
{%
\end{oldthebibliography}%
}%

\begin{document}
\title{Mismatched Estimation of \\Polynomially Damped Signals}

\author{ \IEEEauthorblockN{Filip Elvander, Johan Sw\"ard, and Andreas Jakobsson}
\IEEEauthorblockA{Division of Mathematical Statistics, 
Lund University, Sweden\\
Emails: \{filip.elvander, johan.sward, andreas.jakobsson\}@matstat.lu.se\vspace{-4pt}}
}

%
%\author{\IEEEauthorblockN{Filip Elvander}
%\IEEEauthorblockA{Div. Mathematical Statistics\\
%Lund University, Sweden\\
%Email: filip.elvander@matstat.lu.se}
%\and
%\IEEEauthorblockN{Johan Sw\"ard}
%\IEEEauthorblockA{Div. Mathematical Statistics\\
%Lund University, Sweden\\
%Email: johan.sward@matstat.lu.se}
%\and
%\IEEEauthorblockN{Andreas Jakobsson}
%\IEEEauthorblockA{Div. Mathematical Statistics\\
%Lund University, Sweden\\
%Email: andreas.jakobsson@matstat.lu.se}}
%

\maketitle
\begin{abstract}
In this work, we consider the problem of estimating the parameters of polynomially damped sinusoidal signals, commonly encountered in, for instance, spectroscopy. Generally, finding the parameter values of such signals constitutes a high-dimensional problem, often further complicated by not knowing the number of signal components or their specific signal structures. In order to alleviate the computational burden, we herein propose a mismatched estimation procedure using simplified, approximate signal models. Despite the approximation, we show that such a procedure is expected to yield predictable results, allowing for statistically and computationally efficient estimates of the signal parameters.
\end{abstract}
%%%%%%%%%%%%%%%%%%%%%%%%%%%%
\vspace{2mm}
%%%%%%%%%%%%%%%%%%%%%%%%%%%%
%%%%% Intro %%%%%%%%%%%%%%
%\vspace{-1mm}
\section{Introduction}%\vspace{-2mm}
Signals that may be well modeled as a superposition of exponentially decaying complex-valued sinusoidals appear in a wide variety of fields, such as radar, geology, non-destructive testing, and spectroscopy (see, e.g., \cite{VanhammeSvHvH01,DahlenRJ15_75, AkkeBP93_115}). In this work, we are primarily interested in nuclear magnetic resonance (NMR) spectroscopy, where the signal parameters correspond to properties of the material under study, such as intra-molecular forces. Historically, the model most commonly considered is the so-called Lorentzian line shape \cite{HiginbothamM01_43}, i.e., wherein the decay of the signal components are modeled as an exponential first-degree polynomial, although more detailed signal models are also common, such as the Voigt line shape, which use a second-order polynomial decay \cite{MarshallHBF05_37}.

The estimation of such signal models have been approached in a variety of ways (see \cite{VanhammeSvHvH01} for a more general review), e.g., by exploiting subspace decompositions \cite{VanHuffelCDH94_110,ChanSS12_92}, linear system descriptions \cite{KlintbergMcK19_icassp}, as well as compressed sensing methods \cite{SahnounDSB12_60,SwardAJ16_128,JulhinESJ18_ispacs}. However, for some responses, the first-order polynomial is insufficient for accurately modeling the observed data, and one instead
requires the use of Voigt line shapes to more accurately capture the structure of the signal. This model extension then implies that methods based on linear prediction no longer are applicable, demanding more advanced procedures for estimating the signal parameters.

A common approach for addressing this issue is to form estimates by minimizing a non-linear least squares criterion \cite{BruceHMB00_142}, although such an approach requires prior information about the model order, or that such information is estimated, e.g., by adding sparsity enhancing penalties to the cost function, iteratively solving for one signal component at the time \cite{AdalbjornssonJ10} or performing model order estimation using some predefined criterion \cite{StoicaS04_21}. 
Performing these searches can be computationally cumbersome due to the large parameter space that the signal model entail. To find a remedy for this problem, we propose an estimation procedure that gradually increases the model complexity, such that one is restricting the parameter space over which the parameter search is performed. 
To alleviate this, we use the framework mismatched estimation (see, e.g., \cite{FortunatiGGR17_34} for a recent overview). Specifically, we present a brief analysis on the expected behavior of approximate maximum likelihood estimators (MLEs) derived under simplified model assumptions, i.e., when using lower-order polynomials for describing the signal decay.

In the first part of the paper, we describe the signal model and state the sought optimization problem. After this, we present the expected behavior of mismatched MLEs for simplifications of this model. We then propose a computationally efficient statistical test based on the spectral properties for discriminating between different polynomial decay models. In particular, the proposed test does not require estimates of the parameters corresponding to the more complex signal models, thereby avoiding computationally cumbersome searches for redundant signal parameters. Lastly, we present numerical examples illustrating the performance of the proposed procedure as compared to the Cram{\'e}r-Rao lower bound (CRLB) \cite{AdalbjornssonJ10}.
%
%
% SIGNAL MODEL
\section{Signal Model}
Consider a signal consisting of $K$ polynomially damped sinusoids\footnote{We here restrict our attention to Lorentzian and Voigt line shapes, but note that the presented procedure may also be used for higher order polynomial decays. Such models are used in, for instance, non-destructive testing \cite{DahlenRJ15_75}.}
\begin{align}\label{eq:signal_model}
y_n = \sum_{k=1}^{K}r_ke^{i\phi_k +i\omega_k t_n-\beta_k t_n-\gamma_k t_n^2} + \epsilon_n, \quad 
\end{align}
for $n=1,\dots, N$, where $r_k$ denotes the amplitude, $\phi_k$ the phase, $\omega_k$ the angular frequency, $\beta_k \geq 0$ the Lorentzian damping, and $\gamma_k \geq 0$ the Voigt damping for $k$th signal component. Furthermore, $\epsilon_n$ denotes an additive noise, herein assumed to be well modeled as a circularly symmetric white Gaussian noise\footnote{It may be noted that this constitutes a valid assumption in many spectroscopy applications, with $\epsilon_n$ corresponding to thermal (Johnson) noise.}. 
In this work, we aim to formulate a procedure for estimating the parameters of \eqref{eq:signal_model} in a computationally and statistically efficient manner. It should be stressed that no knowledge of $K$, nor of the number of components of each signal class, i.e., whether certain parameters $\beta_k$ or $\gamma_k$ are strictly positive, is assumed. Thus, the number of signal components, as well as their class, has to be estimated.

As a preliminary, it may be noted that the MLE, assuming knowledge of $K$, may be formed as
\begin{align} \label{eq:mle}
	\argminwrt[\psi_1,\ldots,\psi_K] \sum_{n=0}^{N-1}\abs{ y_n - \sum_{k=1}^K \mu(t_n;\psi_k,\alpha)}^2
\end{align}
where $\psi_k = [r_k, \phi_k,  \omega_k,  \beta_k,  \gamma_k ]$ is the parameter vector for component $k$, and  
\begin{align} 
\mu(t_n;\psi_k,\alpha) = r_k\alpha(t_n;\psi_k)e^{i\phi_k +i\omega_k t_n}
\end{align}
with $\alpha(t_n;\psi_k) = e^{-\beta_k t_n-\gamma_k t_n^2}$ denoting the envelope function. Assuming that the classes of the signal components are unknown, a straightforward approach would be, having estimated all parameters, to conduct a hypothesis test in order to determine whether $\beta_k >0$ and/or $\gamma_k > 0$, for each individual component. However, finding the solution to \eqref{eq:mle} generally requires a $3K$-dimensional search, i.e., over $\omega_k$, $\beta_k$, and $\gamma_k$, as solving for $r_k$ and $\phi_k$ may be done using ordinary least squares. 
As an alternative, one may apply sparse reconstruction techniques by constructing a dictionary over all candidate parameters, reminiscent to the method in \cite{SahnounDSB12_60}, but such a scheme quickly becomes practically infeasible due to the size of the problem for any reasonably fine grid for the different parameters. 
Motivated by mismatched estimation, we here instead formulate a sequential estimation scheme avoiding these difficulties, while still yielding an efficient estimator.

\section{Mismatched estimation}
Consider an observation of the signal in \eqref{eq:signal_model} for the special case when $K = 1$, collected in the vector $\by\!=\!\left[\!\begin{array}{ccc}\!y_1\!&\!\dots\!&\!y_N\!\end{array}\!\right]^T$. The corresponding likelihood is then given by
\begin{align} \label{eq:true_likelihood}
	\mathcal{L}(\by;\psi,\alpha)\!=\!\frac{1}{(\pi\sigma^2)^N}\!\exp\!\left\{\!-\frac{1}{\sigma^2}\!\sum_{n=0}^{N-1}\abs{y_n \!-\!\mu(t_n;\psi,\alpha) }^2\!\right\}
\end{align}
where $\sigma^2$ is the variance of the additive noise. It is worth noting that the likelihood is parametrized by the envelope function $\alpha$, allowing for expressing the difference between the likelihood for cisoid, Lorentzian, and Voigt models solely in terms of $\alpha$. Specifically, we are concerned with the implication of parameter estimation in the case when $\alpha$ is replaced with a misspecified version $\breve \alpha$. We do this by considering the following definition:
\begin{definition}[Pseudo-true parameter\cite{FortunatiGGR17_34}] 
	Consider a signal sample with likelihood $\mathcal{L}$, parametrized by the parameter vector $\psi$. For another parametric likelihood, $\breve{\mathcal{L}}$, parametrized by $\theta$, the pseudo-true parameter, $\theta_0$, is defined as
	\begin{align} \label{eq:pseudo_true}
		\theta_0 = \argminwrt[\theta] -\mathbb{E}_{{\mathcal{L}}}\left( \log \breve{\mathcal{L}}(\by;\theta)  \right).
	\end{align}
\end{definition}
Thus, $\theta_0$ minimizes the Kullback-Leibler divergence between the assumed and true signal models. Interestingly, and of importance for the problem of mismatched estimation, the mismatched likelihood estimator, i.e., 
\begin{align} 
\hat{\theta}_{MLE} = \argmaxwrt[\theta] \log \breve{\mathcal{L}}(\by;\theta)
\end{align}
where $\by$ is sampled from $\mathcal{L}$, converges, under some regularity conditions, to the pseudo-true parameter $\theta_0$ as the signal-to-noise ratio (SNR), or sample size, depending on the application, increases \cite{FortunatiGGR17_34}. Herein, we are interested in estimating the parameters of purely sinusoidal and Lorentzian models when the actual measured signal may be Lorentzian or Voigt, respectively. The following two propositions detail the expected behavior in the one-component case. 
%
%%%%%%%%%% PROPOSITION: sine template
\begin{proposition} \label{prop:sinusoid_template}
	The pseudo-true parameter vector $\theta_0 = [r_0, \omega_0, \phi_0]$ corresponding to a model with $\breve \alpha(t_n;\theta) \equiv 1$, when $\by$ is sampled from \eqref{eq:true_likelihood}, is given by
	\begin{align}
		r_0 = \frac{r}{N} \sum_{n=1}^{N} \alpha(t_n;\psi) \quad,\quad \omega_0 =  \omega \quad,\quad \phi_0 = \phi.
	\end{align}
\end{proposition}
%%%%%%%%% END OF PROPOSITION %%%%%%%%%%%%%%%
%
%
%% PROOOF %%%%%%
\begin{proof}
	Excluding constant terms as well as positive scalings, expectation is given by
	\begin{align}
		2r_0 r \sum_{n=0}^{N-1}\alpha(t_n) \mathfrak{Re}\left( e^{i(\omega_0- \omega)t+i(\phi_0- \phi)}  \right) - N r^2.
	\end{align}
	Here, $\mathfrak{Re}\left( e^{i(\omega_0- \omega)t+i(\phi_0- \phi)}\right) \leq 1$, with equality being achieved for all $t_n$ if and only if $\omega_0 =  \omega$ and $\phi_0 = \phi + k2\pi$, for $k \in \mathbb{Z}$. Substituting this in and minimizing the negative of the resulting expression with respect to $r$ yields the above stated result.
\end{proof}
%%% END PROOOF %%%%%%
%
Thus, the estimate of the frequency and phase parameters are asymptotically unbiased, whereas the amplitude $r$ will be underestimated. Further, for the mismatched estimation of a Lorentzian component, the following proposition holds.
%
%
%%%%% PROPOSITION: BIASED ESTIMATION OF LORENTZIAN %%%%%%
\begin{proposition}\label{prop:lorentzian_template}
Consider estimating the parameters $\theta_0 = [r_0,\phi_0,\omega_0,\beta_0]$ corresponding to a model with $\breve \alpha(t_n;\theta) = e^{-\beta t_n}$, when $\by$ is sampled from \eqref{eq:true_likelihood}. Then, $\omega_0 = \omega$ and $\phi_0 = \phi$. Furthermore,
\begin{align}
	r_0  =  r \frac{\sum_{n=1}^{N} \alpha(t_n;\psi)e^{- \beta_0 t_n}}{\sum_{n=0}^{N-1} e^{-2\beta_0 t_n}},
\end{align}
where the pseudo-true decay parameter satisfies
\begin{align}
	\beta_0 \in \left( \beta ,  \beta + \gamma\left(t_{N}+t_{N-1}\right)\right]	
\end{align}
for any $\gamma > 0$. 
\end{proposition}
%%%%% END PROPOSITION %%%%%%
%
%
%%%%%%%% PROOF %%%%%%%%%%%%
\begin{proof}
The values of $\omega_0$, $\phi_0$, and $r_0$ may be obtained using the same reasoning as in the proof of the previous proposition.
%
%Using the same reasoning as in the proof of the previous proposition, it may be seen that for each value of $\beta$, \eqref{eq:MMLE} is minimized by $\omega = \breve \omega$, $\phi = \breve \phi$, and $r  =  \breve r \frac{\sum_{n=0}^{N-1} e^{-\breve\beta t_n - \breve\gamma t_n^2 - \beta t_n}}{\sum_{n=0}^{N-1} e^{-2\beta t_n}}$.
%
To find $\beta_0$, one needs to maximize
\begin{align*}
	\Lambda(\beta_0) = \frac{1}{2} r \frac{\left( \sum_{n=0}^{N-1} e^{-\beta t_n - \gamma t_n^2 - \beta_0 t_n} \right)^2}{\sum_{n=0}^{N-1} e^{-2\beta_0 t_n}}.
\end{align*}
Differentiation with respect to (w.r.t) $\beta$ yields
\begin{align*}
	\frac{\partial \Lambda(\beta_0)}{\partial \beta_0} = \frac{\lambda_1(\beta_0)+\lambda_2(\beta_0)}{p(\beta_0)}
\end{align*}
where $p$ is a strictly positive function and
\begin{align*}%\nonumber
	\lambda_1(\beta_0) = -2&\left( \sum_{n=0}^{N-1} e^{-\beta t_n - \gamma t_n^2 - \beta_0 t_n} \right)\left( \sum_{n=0}^{N-1} t_n e^{-\beta t_n- \gamma t_n^2 - \beta_0 t_n} \right)\\ 
	&\cdot \left( \sum_{n=0}^{N-1} e^{-2\beta_0 t_n} \right)
\end{align*}
and
\begin{align*}
	\lambda_2(\beta_0) = 2\left( \sum_{n=0}^{N-1} e^{-\beta t_n - \gamma t_n^2 - \beta_0 t_n} \right)^2 \left( \sum_{n=0}^{N-1} t_n e^{-2 \beta_0 t_n} \right).
\end{align*}
Excluding common positive terms, one may arrive at the function (note explicit dependence on $\gamma$)
\begin{align*}
	\Psi(\beta_0,\gamma) &= \left( \sum_{n=0}^{N-1} e^{-\beta t_n - \gamma t_n^2 - \beta_0 t_n} \right) \left( \sum_{n=0}^{N-1} t_n e^{-2\beta_0 t_n} \right) \\
	&\quad -\left( \sum_{n=0}^{N-1} t_n e^{-\beta t_n - \gamma t_n^2 - \beta_0 t_n} \right) \left( \sum_{n=0}^{N-1} e^{-2\beta_0 t_n} \right) \\
	&= \sum_{m = 1}^{N-1}\sum_{n=0}^{m-1}(t_m-t_n)e^{-\beta_0(t_m+t_n)} \\
	&\qquad \cdot \left(e^{- \beta t_n -  \gamma t_n^2- \beta_0 t_m}  -e^{- \beta t_m -  \gamma t_m^2- \beta_0 t_n}\right)
\end{align*}
which has the same zeros w.r.t. to $\beta_0$ and sign as the derivative of the original objective function $\Lambda$. It may be noted that $\Psi( \beta_0, 0) = 0$, i.e., the $\beta_0 =  \beta$ for $ \gamma = 0$. 
However, for $ \gamma > 0$, we will show that $\beta_0 >  \beta$ by showing that $\Psi(\beta_0, \gamma) > 0$ for any $\beta_0 \in [0, \beta]$ and $ \gamma > 0$. In fact, all terms of the last expression of $\Psi(\beta_0, \gamma)$ are positive for $\beta_0 \in [0, \beta]$. To see this, note that $(t_m-t_n)e^{-\beta_0(t_m+t_n)}>0$ for $t_m>t_n$. Thus, positivity is equivalent to
\begin{align*}
	e^{- \beta t_n -  \gamma t_n^2- \beta_0 t_m} > e^{- \beta t_m -  \gamma t_m^2- \beta_0 t_n} \iff \frac{\beta_0 -  \beta}{ \gamma} < t_m +t_n.
\end{align*}
as $\beta_0 -  \beta \leq 0$ for $\beta_0 \in [0, \beta]$, the inequality holds for all $t_m$ and $t_n$. Thus, $\Psi(\beta_0, \gamma) > 0$ for all $\beta_0 \in [0, \beta]$ as it is a sum of strictly positive terms. Using the same line of reasoning, we have that all terms of the sum are non-positive if $\beta_0 \geq  \beta +  \gamma \left( t_{N-1}+t_{N-2} \right)$, i.e., there exists a finite $\beta_0$ such that $\Psi(\beta_0, \gamma)\leq 0$. As $\Psi(\beta_0, \gamma)$ is continuous in $\beta_0$, we may conclude that $\Psi(\beta_0,  \gamma) = 0$ for some $\beta_0 \in \left( \beta ,  \beta + \gamma(t_{N-1}+t_{N-2})\right]$, which concludes the proof.
\end{proof}
%
%%%%%%%% END PROOF %%%%%%%%%%%%
%
%
%
It should be noted that the estimate of the frequency and phase parameters are thus unbiased, even if using an erroneous model, whereas the estimate of the linear decay parameter will incur a strictly positive bias. This result may be used as a bound when forming a final estimate of the parameter $\beta$. It should be stressed that the results of Propositions~\ref{prop:sinusoid_template} and \ref{prop:lorentzian_template} are only exact for the single-component case. However, they will hold approximately for multi-component signals that do not contain components that are too closely spaced in frequency. 

In the following section, we address the problem of how  one may proceed to, in a computationally efficient manner, i.e., without having to form estimates of parameters of more complex models,  determine whether a fitted model is sufficient for describing the measured data.
%
%
 %
%%%%%%%%%%
\section{Spectrum test}
As noted earlier, a Lorentzian model may be obtained as a special case of a Voigt model by setting $\gamma = 0$, and a sinusoidal model may be obtained by setting $\beta = \gamma = 0$. Thus, one may discern between the different models using standard hypothesis tests, i.e., by considering the statistical significance of the signal parameters \cite{AdalbjornssonJ10}. However, such a procedure then requires estimating the full set of signal parameters. In order to avoid fitting unnecessarily complex models, we propose to exploit the difference in spectral properties of the three considered models. As detailed in Propositions~\ref{prop:sinusoid_template} and \ref{prop:lorentzian_template}, we expect that, for signals that do not contain too closely spaced components, the estimates of the frequency and phase parameters will be unbiased. 
Thus, components estimated under mismatched model assumptions will have spectra with the same modes as the actual signal, but with erroneous shapes. Also, in the residual spectra, we expect the power to be concentrated in a neighborhood of the estimated frequencies. 
%This is illustrated in Figure~\ref{fig:residual_spectra}, displaying the spectrum of a Voigt component together with the residual spectra obtained by fitting sinusoidal and Lorentzian models to the signal measurements. As can be seen, the power is indeed mainly concentrated in a neighborhood of the estimated frequency.
With this observation, we propose to decide whether a fitted model is sufficient by considering the whiteness of the residual spectrum using the following proposition (see, e.g., \cite{jakobsson_tsa_19}).
%
%
%%%%% PROP: SPECTRUM AND TEST STATISTIC.
\begin{proposition}[] \label{prop:spectrum_test}
Under the null-hypothesis that the signal template coincides with the measurement model, it holds that the periodogram spectral estimate is distributed according to
\begin{align} \label{eq:per_dist}
	\frac{2\hat{\Phi}^{\mathrm{per}}(\omega)}{\breve \sigma^2} \sim \chi^2(2),
\end{align}
for $ \omega = {k}/{2\pi N}$, with $k =1,2,\ldots,N-1$. Further, the test statistic
\begin{align}\label{eq:test_statistic}
	\xi(\hat{\omega}) \triangleq \frac{\frac{1}{\abs{\ccI(\hat{\omega})}}\sum_{k \in \ccI(\hat{\omega})} \hat{\Phi}^{\mathrm{per}}(\omega_k)}{\frac{1}{N-\abs{\ccI(\hat{\omega})}}\sum_{k \notin \ccI(\hat{\omega})} \hat{\Phi}^{\mathrm{per}}(\omega_k)},
\end{align}
where $\ccI(\hat{\omega})$ is a set of indices such that $\left\{ \omega_k = \frac{k}{N2\pi}\mid k \in \ccI(\hat{\omega}) \right\}$ is distributed according to
\begin{align}
\xi(\hat{\omega}) \sim \mathrm{F}\left( \abs{\ccI(\hat{\omega})}, N - \abs{\ccI(\hat{\omega})} \right),
\end{align}
where $\mathrm{F}(d_1,d_2)$ denotes an F-distribution with $(d_1,d_2)$ degrees of freedom.
\end{proposition}
%%%%%%%%%%%%%%%%%%%%%%%
%
%
\begin{proof}
The proposition follows directly from the results in \cite{jakobsson_tsa_19}, Chapter 5 and the references therein.
\end{proof}
Letting $\ccI(\hat{\omega})$ be a set of indices corresponding to a neighborhood of the estimated frequency $\hat{\omega}$, one may thus use Proposition~\ref{prop:spectrum_test} for detecting residual spectral power by comparing the test statistic $\xi(\hat{\omega})$ to quantiles of the F-distribution. 
The benefit of this approach, as compared to standard hypothesis testing, is that one does not have to estimate the parameters of the more complex model in order to decide whether it is needed or not; all that is required is a candidate, potentially mismatched, the model, and an estimated spectrum. For simplicity, we here consider the periodogram estimate, although one could envision more sophisticated alternatives, although the distribution of the test statistic may in such case be different.
In practice, for the case of multi-component signals, the set of spectral points used to estimate the noise power may be taken to be the complement of the union of all neighborhoods $\ccI(\hat{\omega})$, as opposed to the complement of the single neighborhood $\ccI(\hat{\omega})$ in Proposition~\ref{prop:spectrum_test}.
%
%
%
%
%
%
%%%%%%%%%%%%% Figure: EMPIRICAL CDF, OMEGA %%%%%%
\begin{figure}[t!]
        \centering
        %\vspace{1mm}
        %\vspace{-2.5 mm}
            \includegraphics[width=0.45\textwidth]{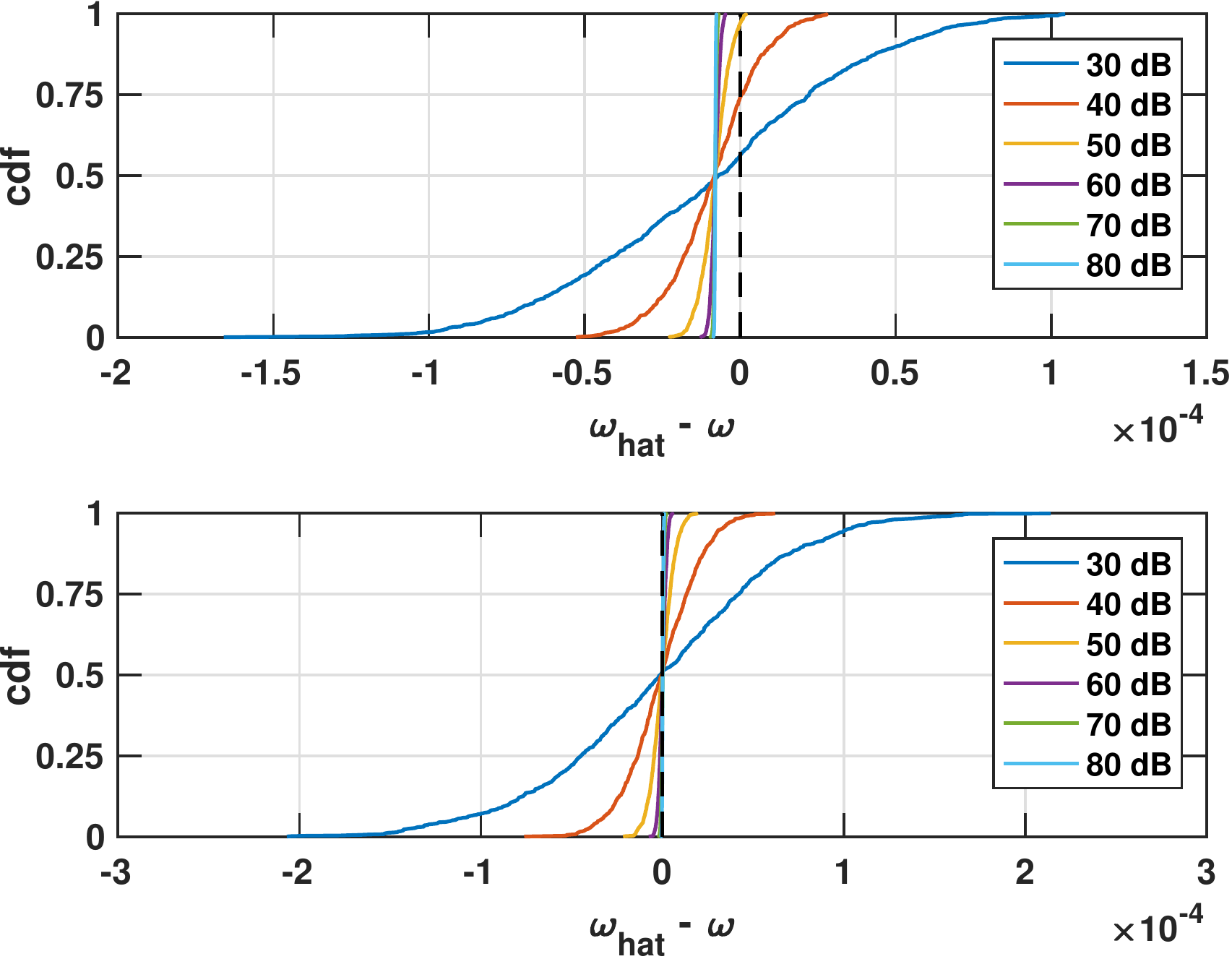}
            %\vspace{.75mm}
           \caption{Empirical distribution functions for the error of initial frequency estimates. Top: Lorentzian component. Note that the Lorentzian and sine components are closely spaced. Bottom: Voigt component. Note that the Voigt component is well separated for the other two signal components.}
            \label{fig:empirical_cdf_omega}
%\vspace{-2mm}
\end{figure}
%%%%%%%%%%%%%%%%%%%%%%%%%%%%%%%%%%
%
%
%
%
%
%
%
\section{Proposed algorithm}
Based on the results in the previous section, we propose to address the estimation problem using the following procedure\footnote{The further steps are here only taken if the currently used model at each step is deemed to have been insufficient.}
\begin{enumerate}
	\item Fit a purely sinusoidal model to the signal of interest.
	\item Identify the damped components using Proposition~\ref{prop:spectrum_test}. Add the identified components to signal residual.
	\item Fit a Lorentzian model to the signal residual.
	\item Identify the non-linearly damped components using Proposition~\ref{prop:spectrum_test}. Add these  to the signal residual.
	\item Fit a Voigt model to the signal residual. 
\end{enumerate}
For the first and third steps, we propose to use the estimators SURE-IR \cite{FangWSLB16_64} and D-SURE \cite{JulhinESJ18_ispacs}. These estimators are approximate ML estimators tailored for sinusoidal and Lorentzian models, respectively, and both include techniques for arriving at sparse estimates not requiring \textit{a priori} model order knowledge. Also, by not relying on a fixed gridding of the parameter space, the need for large signal dictionaries is alleviated, allowing for rapid implementation. In practice, the complexity of these estimators are $\mathcal{O}(\hat{K}^3)$, where $\hat{K}$ is the number of signal components identified by the estimators. Clearly, steps 1 and 3 may also be implemented using other estimators, such as, e.g., root-MUSIC \cite{RaoH89_37} and its extensions \cite{VanHuffelCDH94_110,ChanSS12_92}, although such methods in general require accurate \textit{a priori} model order knowledge. In any case, in order to utilize the results from Propositions~\ref{prop:sinusoid_template} and \ref{prop:lorentzian_template}, the chosen estimators should approximate the MLE for the respective models.

%
%
%
%
%%%%%%%%%%%%% Figure: SIMULATION, PROBABILITY %%%%%%
\begin{figure}[t!]
        \centering
        %\vspace{1mm}
        %\vspace{-2.5 mm}
            \includegraphics[width=0.45\textwidth]{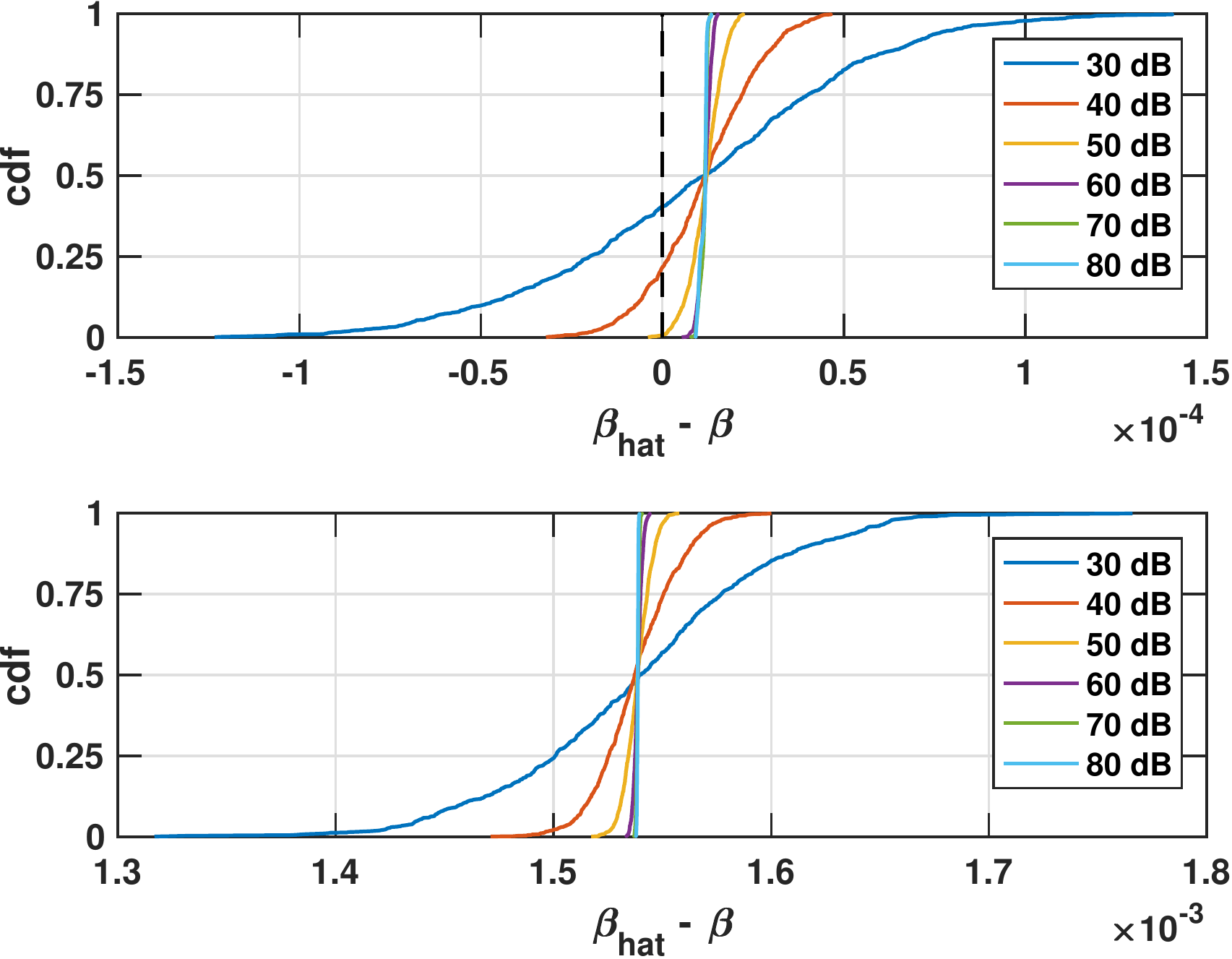}
            %\vspace{.75mm}
           \caption{Empirical distribution functions for the error of initial linear decay estimates. Top: Lorentzian component. Bottom: Voigt component. Note that the bound for the estimation bias is $4 \times 10^{-3}$.}
            \label{fig:empirical_cdf_beta}
%\vspace{-2mm}
\end{figure}
%%%%%%%%%%%%%%%%%%%%%%%%%%%%%%%%%%
%

As noted earlier, one could, as an alternative to the proposed sequential estimation, instead form a full dictionary that spans the whole parameter space and thus contains a (large) set of signal candidates and estimate the parameters by solving a sparsity enforcing optimization problem, e.g. Lasso \cite{HastieTW15,SwardAJ16_128}. However, this approach has the major drawback of requiring a large number of candidate signal components in order to form reliable estimates. For instance, if one is to form the dictionary with $P$ grid point for each parameter, the resulting dictionary would be of the size $N\times P^3$. To yield a satisfactory estimation precision for each parameter, $P$ generally has to be quite large, often in the range of $10^{2} - 10^{4}$. Solving problems on this scales requires significant computational resources and ignores available information.
%

% and is potentially wasteful since it does not take into account the properties described above.
%
%
In contrast, steps 1-4 above allows for an efficient estimation of the parameters of sinusoidal and Lorentzian components, as well as identifying Voigt components. Using Proposition~\ref{prop:lorentzian_template}, one may then use a non-linear search to efficiently estimate the Voigt parameters, as the proposition allows a limit on the relevant search space. As a final step, a local refinement search in a limited neighborhood of the estimated parameters may  be performed. It should be noted that the proposed scheme allows for the dimension of the estimate to be kept to a minimum in every step by sequentially identifying and estimating the signal components.

%
%%
%
%
%%%%%%%%%%%%% Figure: SIMULATION, PROBABILITY + RMSE FOR OMEGA %%%%%%
\begin{figure}[t!]
        \centering
        %\vspace{1mm}
        %\vspace{-2.5 mm}
            \includegraphics[width=0.45\textwidth]{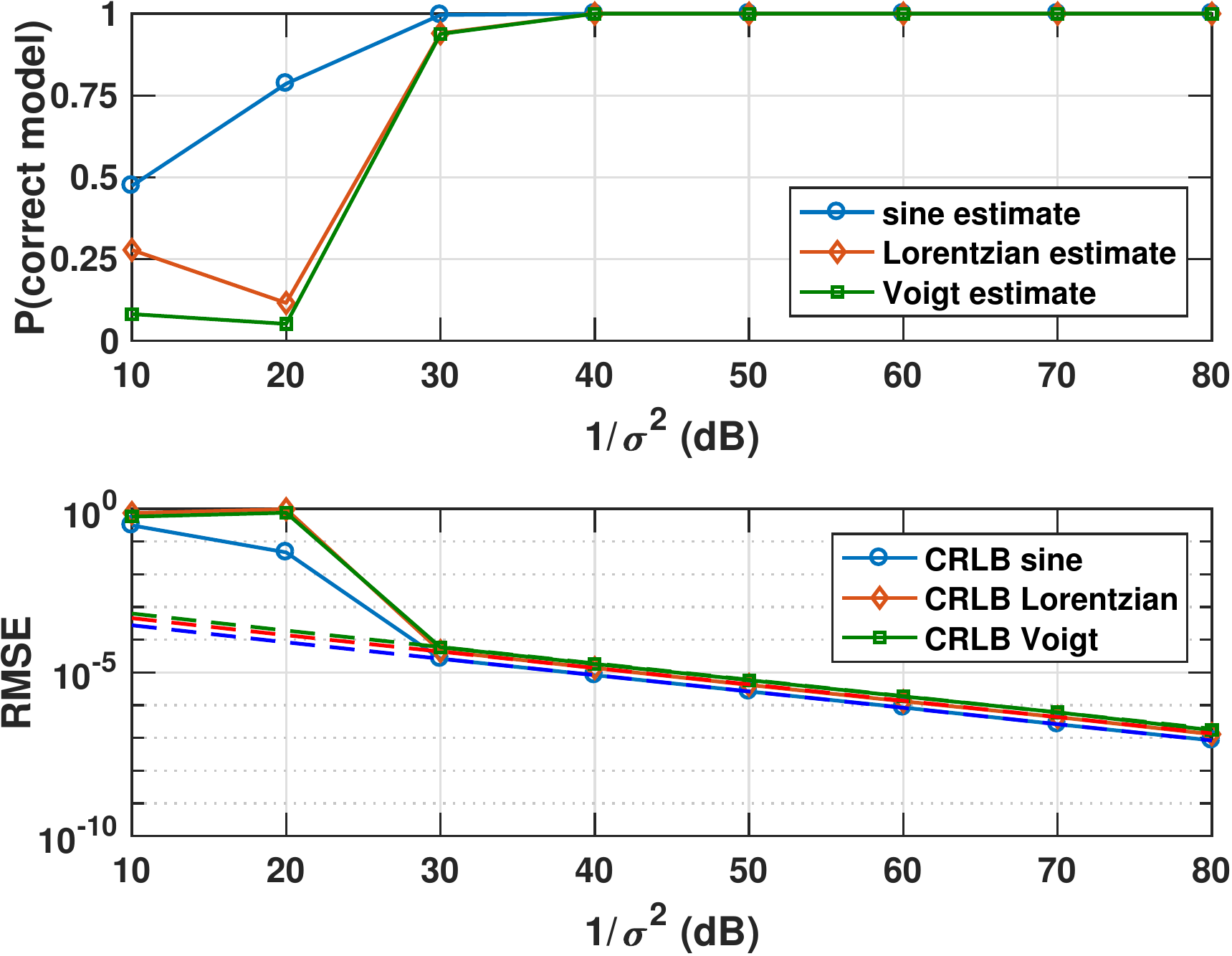}
            %\vspace{.75mm}
           \caption{Top: Empirical probability of correctly classifying the three signal components. Bottom: Root-MSE for the frequency parameters, $\omega_k$.}
            \label{fig:simulation_probability_separated}
%\vspace{-2mm}
\end{figure}
%%%%%%%%%%%%%%%%%%%%%%%%%%%%%%%%%%
%
%
%
%
%
%
%\vspace{-8mm}
\section{Numerical illustration}
We consider a signal consisting of a pure sinusoidal, a Lorentzian, and a Voigt component, measured at $t_n = n-1$, for $n = 1, 2,\ldots,N$, with $N = 200$.
The signal parameters are $(\omega_1,\omega_2,\omega_3) = (0.7, 0.5, 1.5)$, $(\beta_2, \beta_3) = (\frac{1}{200},\frac{1}{150})$, and $\gamma_3 = 10^{-5}$, where the indices 1, 2, and 3 correspond to the sinusoidal, Lorentzian, and Voigt components, respectively, and $r_1 = r_2 = r_3 = 1$, with the phases drawn uniformly on $[0,2\pi)$.
We add circularly symmetric white Gaussian noise, and attempt to recover the signal parameters using the proposed estimation method. We repeat this procedure in 500 Monte Carlo simulations for different levels of noise.\footnote{We stress that the estimator is not provided with oracle knowledge of the number of signal components, nor their specific signal category.}
Figure~\ref{fig:empirical_cdf_omega} displays the empirical cumulative distribution function %(CDF)
for the error of the frequency estimates for the Lorentzian and Voigt components obtained in Step 1, i.e., when fitting a purely sinusoidal model.  There is a small negative estimation bias for the Lorentzian component, caused by its proximity to the sine component, whereas the estimate for the Voigt component, that is well-separated from the other two components, is unbiased.
Figure~\ref{fig:empirical_cdf_beta} displays corresponding results for the linear decay parameters obtained in Step 3, i.e., when fitting Lorentzian components. The estimate corresponding to the Voigt component has a substantial bias, given the scale of the parameters, although it is within the upper bound $4\times 10^{-3}$ predicted by Proposition~\ref{prop:lorentzian_template}.
%
%
%%%%%%%%%%%%% Figure: SIMULATION, RMSE FOR BETA %%%%%%
\begin{figure}[t!]
        \centering
        %\vspace{1mm}
        %\vspace{-2.5 mm}
            \includegraphics[width=0.45\textwidth]{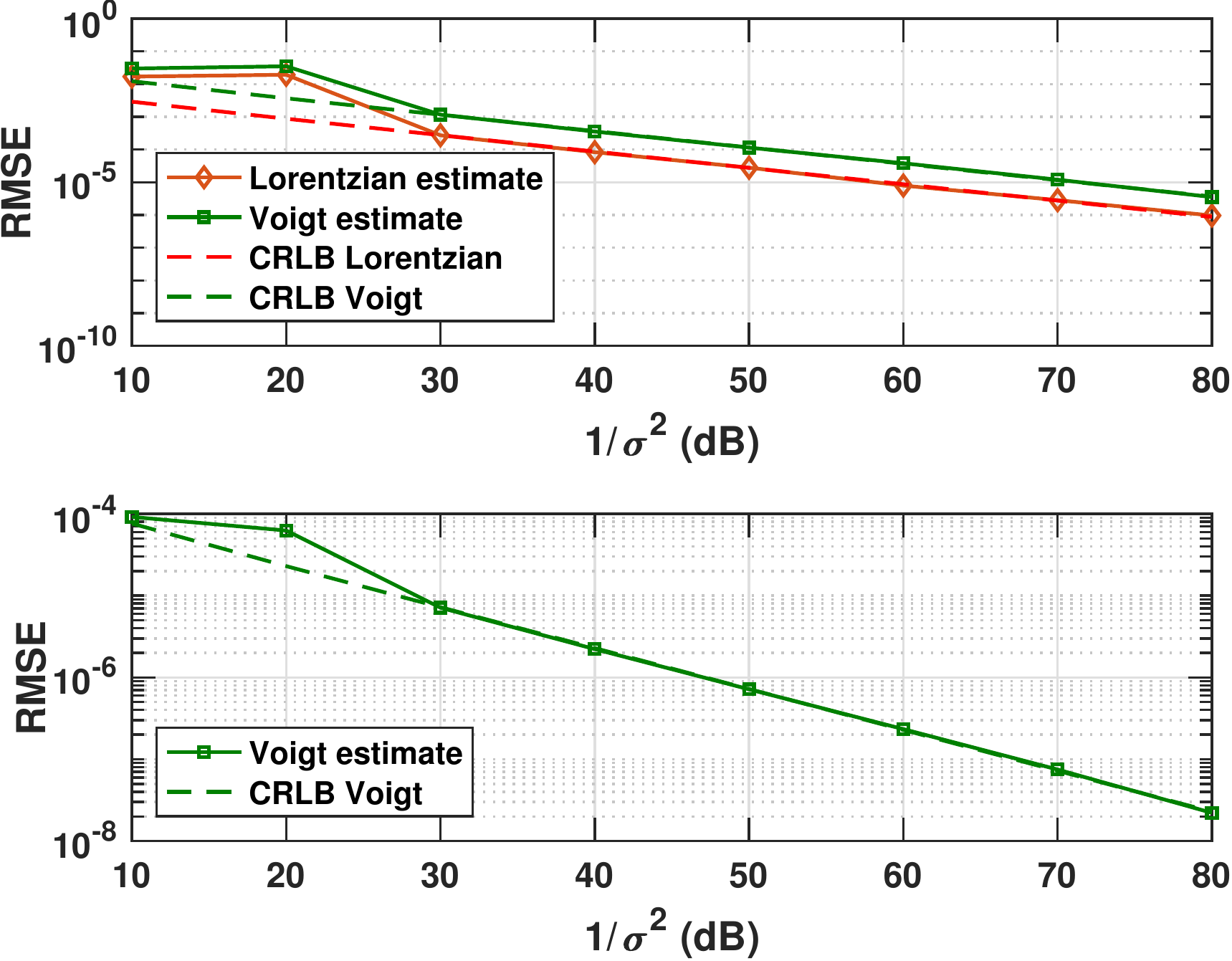}
            %\vspace{.75mm}
           \caption{Top: Root-MSE for the linear decay parameters, $\beta_k$. Bottom: Root-MSE for the quadratic decay parameter, $\gamma$.}
            \label{fig:simulation_beta_rmse_separated}
%\vspace{-2mm}
\end{figure}
%%%%%%%%%%%%%%%%%%%%%%%%%%%%%%%%%%
%
%
The empirical probability of correct model classification, i.e., the probability of determining $K = 3$ as well as correctly classifying the model type for each component, is displayed in the top panel of Figure~\ref{fig:simulation_probability_separated}. As can be seen, the probability approaches 1, for all three components, as the noise variance falls below $10^{-3}$. Considering the simulations in which the components were correctly classified, the obtained root-MSE, being refined with a local non-linear search, % in the neighborhood of the initial estimates, 
for the signal parameters is shown in the bottom panels of Figure~\ref{fig:simulation_probability_separated} and \ref{fig:simulation_beta_rmse_separated}, as compared with the corresponding CRLB.
%

%\newpage
\bibliographystyle{IEEEbib}
\bibliography{IEEEabrv,ElvanderSJ19_camsap_arXiv.bbl}

\begin{thebibliography}{10}

\bibitem{VanhammeSvHvH01}
L.~Vanhamme, T.~Sundin, P.~van Hecke, and S.~van Huffel,
\newblock ``{MR} spectroscopy quantitation: a review of time-domain methods,''
\newblock {\em NMR in Biomedicine}, vol. 14, no. 4, pp. 233--246, 2001.

\bibitem{DahlenRJ15_75}
U.~Dahl{\'e}n, N.~Ryd{\'e}n, and A.~Jakobsson,
\newblock ``{D}amage {I}dentification in {C}oncrete using {I}mpact
  {N}on-{L}inear {R}everberation {S}pectroscopy,''
\newblock {\em NDT \& E International}, vol. 75, pp. 15--25, 2015.

\bibitem{AkkeBP93_115}
M.~Akke, R.~Br{\"u}schweiler, and A.~G. Palmer,
\newblock ``{NMR} order parameters and free energy: an analytical approach and
  its application to cooperative {C}a2+ binding by calbindin-{D}(9{K}),''
\newblock {\em J. Am. Chem. Soc.}, vol. 115, pp. 9832--9833, 1993.

\bibitem{HiginbothamM01_43}
J.~Higinbotham and I.~Marshall,
\newblock ``{NMR} {L}ineshapes and {L}ineshape {F}itting {P}rocedures,''
\newblock {\em Annu. Rep. NMR Spectrosc.}, vol. 43, pp. 59--120, 2001.

\bibitem{MarshallHBF05_37}
I.~Marshall, J.~Higinbotham, S.~Bruce, and A.~Freise,
\newblock ``{U}se of {V}oigt lineshape for quantification of in vivo 1{H}
  spectra,''
\newblock {\em Magn. Res. in Medicine}, vol. 37, pp. 651--657, 2005.

\bibitem{VanHuffelCDH94_110}
S.~Van Huffel, H.~Chen, C.~Decanniere, and P.~Van Hecke,
\newblock ``{A}lgorithm for {T}ime-{D}omain {NMR} {D}ata {F}itting based on
  {T}otal {L}east {S}quares,''
\newblock {\em J. Magn. Reson.}, vol. 110, pp. 228--237, 1994.

\bibitem{ChanSS12_92}
F.~K.~W. Chan, H.~C. So, and W.~Sun,
\newblock ``{S}ubspace approach for two-dimensional parameter estimation of
  multiple damped sinusoids,''
\newblock {\em Signal Process.}, vol. 92, pp. 2172 -- 2179, 2012.

\bibitem{KlintbergMcK19_icassp}
J.~Klintberg and T.~McKelvey,
\newblock ``An {I}mproved {M}ethod for {P}arametric {S}pectral {E}stimation,''
\newblock in {\em Proc. 44th IEEE Int. Conf. on Acoustics, Speech, and Signal
  Processing}, Brighton, UK, May 13-17 2019, pp. 5551--5555.

\bibitem{SahnounDSB12_60}
S.~Sahnoun, E.~Djermoune, C.~Soussen, and D.~Brie,
\newblock ``{S}parse multidimensional modal analysis using a multigrid
  dictionary refinement,''
\newblock {\em EURASIP J. Applied SP}, vol. 60, pp. 1--10, 2012.

\bibitem{SwardAJ16_128}
J.~Sw\"ard, S.~I. Adalbj\"ornsson, and A.~Jakobsson,
\newblock ``{H}igh {R}esolution {S}parse {E}stimation of {E}xponentially
  {D}ecaying {N}-dimensional {S}ignals,''
\newblock {\em Elsevier Signal Processing}, vol. 128, pp. 309--317, Nov 2016.

\bibitem{JulhinESJ18_ispacs}
M.~Juhlin, F.~Elvander, J.~Sw\"ard, and A.~Jakobsson,
\newblock ``Fast {G}rid-less {E}stimation of {D}amped {M}odes,''
\newblock in {\em 2018 Int. Symp. on Intel. Sig. Proc. and Com. Syst.},
  Ishigaki Island, Okinawa, Japan, November 2018.

\bibitem{BruceHMB00_142}
S.~D. Bruce, J.~Higinbotham, I.~Marshall, and P.~H. Beswick,
\newblock ``{A}n {A}nalytical {D}erivation of a {P}opular {A}pproximation of
  the {V}oigt {F}unction for {Q}uantification of {NMR} {S}pectra,''
\newblock {\em J. Magn. Reson.}, vol. 142, no. 1, pp. 57 -- 63, 2000.

\bibitem{AdalbjornssonJ10}
S.~I. Adalbj\"{o}rnsson and A.~Jakobsson,
\newblock ``{R}elax-{B}ased {E}stimation of {V}oigt {L}ineshapes,''
\newblock in {\em 18th European Signal Processing Conference, EUSIPCO 2010},
  Aalborg, Denmark, 2010, pp. 1053--1057.

\bibitem{StoicaS04_21}
P.~Stoica and Y.~Sel\'en,
\newblock ``{M}odel-order {S}election --- {A} {R}eview of {I}nformation
  {C}riterion {R}ules,''
\newblock {\em {IEEE} Signal Process. Mag.}, vol. 21, no. 4, pp. 36--47, July
  2004.

\bibitem{FortunatiGGR17_34}
S.~{Fortunati}, F.~{Gini}, M.~S. {Greco}, and C.~D. {Richmond},
\newblock ``Performance bounds for parameter estimation under misspecified
  models: Fundamental findings and applications,''
\newblock {\em IEEE Signal Processing Mag.}, vol. 34, no. 6, pp. 142--157, Nov
  2017.

\bibitem{jakobsson_tsa_19}
A.~Jakobsson,
\newblock {\em {T}ime {S}eries {A}nalysis and {S}ignal {M}odeling},
\newblock Studentlitteratur AB, Sweden, 3 edition, 2019.

\bibitem{FangWSLB16_64}
J.~Fang, F.~Wang, Y.~Shen, H.~Li, and R.~S. Blum,
\newblock ``{S}uper-{R}esolution {C}ompressed {S}ensing for {L}ine {S}pectral
  {E}stimation: {A}n {I}terative {R}eweighted {A}pproach,''
\newblock {\em IEEE Trans. Signal Process.}, vol. 64, no. 18, pp. 4649--4662,
  September 2016.

\bibitem{RaoH89_37}
B.~D. Rao and K.~V.~S. Hari,
\newblock ``{P}erformance {A}nalysis of {R}oot-{MUSIC},''
\newblock {\em {IEEE} Trans. Acoust., Speech, Signal Process.}, vol. 37, no.
  12, pp. 1939--1949, December 1989.

\bibitem{HastieTW15}
T.~Hastie, R.~Tibshirani, and M.~Wainwright,
\newblock {\em Statistical Learning with Sparsity: The Lasso and
  Generalizations},
\newblock Chapman and Hall/CRC, 2015.

\end{thebibliography}
\end{document}